\documentclass[11pt,a4paper]{article}
\usepackage{a4wide}
\usepackage{amsmath}
\usepackage{amsthm}

\usepackage[ocgcolorlinks]{hyperref} % Option ocgcolorlinks makes links only colored on
\usepackage{cleveref}

\usepackage{url}
\usepackage{epsfig}
\usepackage{psfrag}
\usepackage{amssymb}
\usepackage{wrapfig}
\usepackage{algorithm}
\usepackage{algorithmic}
\usepackage{paralist}
\usepackage{xcolor}
\usepackage[normalem]{ulem}
\usepackage{mathptmx}
\usepackage[scaled]{helvet}   %% Takeshi comment out three lines tempolary
\usepackage{courier}

\usepackage[T1]{fontenc}
\usepackage[utf8]{inputenc}

\colorlet{DarkRed}{red!60!black}
\colorlet{DarkGreen}{green!50!black}
\colorlet{DarkBlue}{blue!80!black}
\colorlet{DarkMagenta}{magenta!80!black}

\hypersetup{
	linkcolor = DarkRed,
	citecolor = DarkGreen,
	urlcolor = DarkBlue,
	bookmarksnumbered = true,
	linktocpage = true
}

\newcommand{\supp}{\mathrm{supp}}
\newcommand{\cost}{\mathrm{cost}}

\newcommand{\norm}[1]{\| #1 \|}
\newcommand{\onenorm}[1]{\norm{#1}_1}

\newcommand{\Z}{\mathbb{Z}}
\newcommand{\R}{\mathbb{R}}
\newcommand{\abs}[1]{| #1 |}
\newcommand{\lref}[1]{(\ref{#1})}
\newcommand{\set}[2]{\left\{\, #1 \,:\, #2 \,\right\}}

\newtheorem{theorem}{Theorem}
\newtheorem{lemma}{Lemma}
\newtheorem{corollary}{Corollary}
\newcommand{\ignore}[1]{}

\DeclareMathOperator{\argmin}{arg\,min}

\begin{document}

% \tableofcontents

\title{Convergence of the Non-Uniform Physarum Dynamics}
\author{Andreas Karrenbauer\thanks{Max Planck Institute for Informatics, Saarland Informatics Campus.} \and Pavel Kolev$^*$ \and Kurt Mehlhorn$^*$}

\maketitle

\begin{abstract} The Physarum computing model is an analog computing model motivated by the network dynamics of the slime mold Physarum Polycephalum. In previous works, it was shown that it can solve a class of linear programs. We extend these results to a more general dynamics motivated by situations where the slime mold operates in a non-uniform environment. 
  
Let $c \in \Z^m_{> 0}$, $A \in \Z^{n\times m}$, and $b \in \Z^n$. We show under fairly general conditions that the non-uniform Physarum dynamics
  \[       \dot{x}_e = a_e(x,t) \left(\abs{q_e} - x_e\right) \]
  converges to the optimum solution $x^*$ of the weighted basis pursuit problem minimize $c^T x$ subject to $A f = b$ and $\abs{f} \le x$. Here, $f$ and $x$ are $m$-dimensional vectors of real variables,  $q$ minimizes the energy $\sum_e (c_e/x_e) q_e^2$ subject to the constraints $A q = b$ and $\supp(q) \subseteq \supp(x)$, and $a_e(x,t) > 0$ is the reactivity of edge $e$ to the difference $\abs{q_e} - x_e$ at time $t$ and in state $x$. Previously convergence was only shown for the uniform case $a_e(x,t) = 1$ for all $e$, $x$, and $t$.

  We also show convergence for the dynamics
  \[ \dot{x}_e = x_e  \left( g_e \left(\frac{\abs{q_e}}{x_e}\right) - 1\right),\]
where each $g_e$ is an increasing differentiable function with $g_e(1) = 1$ (satisfying some mild conditions). Previously, convergence was only shown for the special case of the shortest path problem on a graph consisting of two nodes connected by parallel edges. 
\end{abstract}

\section{Introduction}

The \emph{Physarum computing model} is an analog computing model motivated by the network dynamics of the slime mold Physarum polycephalum.  In wet-lab experiments, it was observed that the slime mold is apparently able to solve shortest path problems~\cite{Nakagaki:2000}. A mathematical model for the dynamic behavior of the slime was proposed in~\cite{Tero-Kobayashi-Nakagaki}. It models the slime network as an electrical network with time-varying resistors that react to the amount of electrical current flowing through them. A more general model for the dynamics was introduced in~\cite{PhysarumMinimumRiskPath} to deal with situations in which the slime has to operate in a non-uniform environment. This more general dynamics is the subject of this paper.  In Section~\ref{Biological Background}, we give more details on the biological background and also survey the theoretical work on the Physarum dynamics. 

The \emph{weighted basis pursuit problem} asks to find the  minimal weighted one-norm solution of a linear system. Formally, 
\begin{equation}\label{LP}
  \text{minimize }   c^T x\text{ subject to } A f = b \text{ and }\abs{f} \le x,
  \end{equation}
  where $c \in \Z^m_{> 0}$, $A \in \Z^{n\times m}$, $b \in \Z^n$, and $x$ and $f$ are $m$-dimensional vectors of real variables. The absolute-value operator is applied componentwise. The matrix $A$ is assumed to have full row-rank; this implies $n \le m$. For simplicity, we also assume that any two basic feasible solutions of $Af = b$ have distinct cost\footnote{A basic feasible solution of $Af = b$ has the form $f = (f_B,f_{\overline{B}})$, where $B$ is a subset of $[m]$ of size $n$, $\overline{B} = [m] \setminus B$,  the submatrix $A_B$ of $A$ is invertible, $f_B = A_B^{-1}b$, and $f_{\overline{B}} = 0$. The cost of such a solution is $c^T \abs{f}$.}; in particular, the optimal solution $(f^*,x^*)$  to~\lref{LP} is unique. We index the rows of $A$ by $i$ and the columns of $A$ by $e$ and, for historical reasons (see Subsection~\ref{Shortest Paths}), refer to the rows as nodes and the columns as edges. 

The \emph{Physarum dynamics} evolves a vector $x(t) \in \R^m_{\ge 0}$ according to the dynamics
\begin{equation}\label{PD}   \dot{x} = \abs{q} - x, \end{equation}
where $q$ is the minimum energy feasible solution of $Af = b$ according to the resistances $r_e = c_e/x_e$:
\[   q = \argmin_f \set{\sum_e (c_e/x_e) f_e^2}{A f = b \text{ and } f_e = 0 \text{ whenever } x_e = 0}.\]
The Physarum dynamics was introduced by biologists~\cite{Tero-Kobayashi-Nakagaki} as a model of the behavior of the slime mold Physarum polycephalum. We discuss the biological background in Section~\ref{Biological Background}. 

In~\cite{BeckerBonifaciKarrenbauerKolevMehlhorn} and~\cite{Facca-Cardin-Putti} it was shown that the Physarum dynamics~\lref{PD} can solve the \emph{weighted basis pursuit problem}~\lref{LP}.

\begin{theorem}[\cite{BeckerBonifaciKarrenbauerKolevMehlhorn,Facca-Cardin-Putti}] Assume a strictly positive starting vector $x(0) \in R^m_{> 0}$. Then the solution $x(t)$ to the dynamics~\lref{PD} is defined on $[0,\infty)$, $x(t) > 0$ for all $t$, and $x(t)$ and $\abs{q(t)}$ converge to the optimal solution $x^*$ of~\lref{LP}. \label{uniform convergence}
\end{theorem}

Actually, the papers~\cite{BeckerBonifaciKarrenbauerKolevMehlhorn,Facca-Cardin-Putti} show convergence under the more general condition $c \ge 0$ and $c^T \abs{f} > 0$ for any $f$ in the kernel of $A$, but here we do not need this generality.

In this paper, we consider the more general dynamics \lref{NUPD} and~\lref{RefinedDynamics}. In the dynamics~\lref{NUPD}, the edges react with different speed to differences between minimum energy solution and capacity, i.e., 
\begin{equation}\label{NUPD}   \dot{x}_e = a_e(x,t)  \left(\abs{q_e} - x_e\right),  \end{equation}
where $a_e(x,t) \ge 0$ is the reactivity of edge $e$ at time $t$, i.e., the edges no longer react uniformly to differences between $\abs{q}$ and $x$, but the reactivity depends on the edge, the current state, and the time. We refer to  \lref{NUPD} as the \emph{non-uniform Physarum dynamics}. The special case that $a_e(x,t)$ is a positive constant for each edge was introduced in~\cite{PhysarumMinimumRiskPath} to model the behavior of Physarum polycephalum in non-uniform environments; see Subsection~\ref{Minimum Risk Paths}.

In Section~\ref{Proofs}, we prove our main technical contribution for the dynamics \lref{NUPD}:

\begin{theorem}\label{main theorem}
  Assume $x(0) > 0$, $0 \le a_e(x,t) \le C$ for all $e$, $x$, and $t$ and some constant $C$, and $a_e(x,t)$ is Lipschitz-continuous. Then: 
  \begin{compactenum}[(a)]
  \item The dynamics \lref{NUPD} has a unique solution $x(t) > 0$ for $t \in [0,\infty)$.
  \item The function
    \begin{equation}\label{Lyapunov Function A}
      L(x,t) = \sum_e (c_e/x_e) q_e^2 + \sum_e c_e x_e\end{equation}
    is a Lyapunov function for the dynamics \lref{NUPD}, i.e., $\frac{d}{dt} L(x,t) \le 0$ for all $t \in [0,\infty)$. Moreover, $\frac{d}{dt} L(x,t) = 0$ if and only if for all $e$: either $a_e(x,t) = 0$ or $x_e(t) = 0$ or $\abs{q_e} = x_e$.
  \item If, in addition, $a_e(x,t) \ge \epsilon$ for some positive $\epsilon$ and all $e$, $x$, and $t$,
 then $\bar{x} \ge 0$ is a fixed point of~\lref{NUPD} if and only if $\bar{x} = \abs{f}$ for a basic feasible solution of $Af = b$.
  \item Under the same additional assumption as in (c), $x(t)$ and $\abs{q(t)}$ converge to a fixed point of \lref{NUPD} as $t$ goes to infinity. 
    \item If, in addition, $a_e(x,t)$ does not depend on $x$ and $\frac{d}{dt}{a_e}(t) \le 0$ for all $e$ and $t$, then $x(t)$ and $\abs{q(t)}$ converge to $x^*$ as $t$ goes to infinity. In particular, this holds true if $a_e(t)$ is a positive constant for all $e$.
  \end{compactenum}
\end{theorem}

The proof of part (a) is standard and part (c) was shown in~\cite{BeckerBonifaciKarrenbauerKolevMehlhorn}. The Lyapunov function in part (b) was introduced in~\cite{Facca-Daneri-Cardin-Putti}. In ~\cite{Facca-Cardin-Putti} it was shown to be a Lyapunov function for the uniform case, i.e., $a_e(x,t) = 1$ for all $e$, $x$, and $t$. We observe that the Lyapunov function also works for the non-uniform dynamics. Part (d) follows easily from parts (b) and (c). Finally, the proof of part (e) is inspired by~\cite{BeckerBonifaciKarrenbauerKolevMehlhorn}.

The function $L(x,t)$ is a Lyapunov function for the Physarum dynamics under very general conditions. Essentially, the only requirement is that $\dot{x}_e$ has the same sign as $\abs{q_e} - x_e$.
For the existence of a solution with domain $[0,\infty)$, we also need that $\dot{x}_e/(\abs{q_e} - x_e)$ is bounded. For the convergence to a fixed point, we need in addition that $\dot{x}_e/(\abs{q_e} - x_e)$ is bounded away from zero.

In the dynamics~\lref{RefinedDynamics}, each edge has its own transfer function that determines how it  reacts to the ratio of flow and capacity being larger or smaller than one, i.e.,
\begin{equation}\label{RefinedDynamics}
  \dot{x}_e = x_e  \left( g_e\left( \frac{\abs{q_e}}{x_e}\right) - 1 \right) \quad \text{for all $e \in [m]$}, \end{equation}
where the response function $g_e: \R_{\ge 0} \rightarrow \R_{\ge 0}$ is assumed to be an increasing differentiable function satisfying $g_e(1) = 1$. Bonifaci introduced this model in~\cite{Bonifaci-RefinedModel} in order to deal with the larger class of response functions proposed in the biological literature. For the shortest path problem in a network of parallel links\footnote{The shortest path problem is a min-cost flow problem where we want to send one unit of flow between two distinguished nodes. For the case of parallel links, the graph has exactly two nodes and all edges run between these nodes.}, \cite{Bonifaci-RefinedModel} shows convergence to the shortest path. Bonifaci assumes the same response function for every edge, but his proof actually works for response functions depending on the edge. 

In Section~\ref{Bonifaci's Refined Model}, we prove our main technical contribution for the dynamics \lref{RefinedDynamics}:

\begin{theorem}\label{Bonifaci main theorem}
  Assume $g_e: \R_{\ge 0} \rightarrow \R_{\ge 0}$ is an  increasing and differentiable function
  satisfying $g_e(1) = 1$, for all $e$. Then,
  \begin{compactenum}[(a)]
  \item The dynamics \lref{RefinedDynamics} has a unique solution $x(t) > 0$ for $t \in [0,\infty)$.
    \item $x \ge 0$ is a fixed point of~\lref{RefinedDynamics} if $x = \abs{f}$ for a basic feasible solution of~\lref{LP}. 
  \item The function
    \begin{equation}\label{Lyapunov Function B}
      L(x,t) = \sum_e (c_e/x_e) q_e^2 + \sum_e c_e x_e\end{equation}
    is a Lyapunov function for the dynamics \lref{RefinedDynamics}, i.e., $\frac{d}{dt} L(x,t) \le 0$ for all $t \in [0,\infty)$. Moreover, $\frac{d}{dt} L(x,t) = 0$ if and only if $x$ is a fixed point.
    \item $x(t)$ and $\abs{q(t)}$ converge to a fixed point of \lref{RefinedDynamics}. 
    \item If, in addition, $g_e(y) \ge 1 + \alpha(y - 1)$ for some $\alpha > 0$ and all $e$ and $y$,
    then $x(t)$ and $\abs{q(t)}$ converge to $x^*$ as $t$ goes to infinity. 
  \end{compactenum}
\end{theorem}
The proof of part (a) is standard and part (b) was shown in~\cite{BeckerBonifaciKarrenbauerKolevMehlhorn}. The Lyapunov function in part (c) was introduced in~\cite{Facca-Daneri-Cardin-Putti}. We observe that it also applies to the dynamics~\lref{RefinedDynamics}. Part (d) follows easily from part (c). Finally, the proof of part (e) is inspired by~\cite{BeckerBonifaciKarrenbauerKolevMehlhorn}. Theorem~\ref{Bonifaci main theorem}
also holds when the function $g_e$ depends on the time and the state. 

Nature does not compute exactly, i.e., one should not expect that in a biological system $\dot{x}_e$ is exactly equal to $\abs{q_e} - x_e$ or to $x_e ( g_e(\abs{q_e}/x_e) - 1 )$. Rather, there will be a noise. Our results show that the dynamics~\lref{NUPD} and~\lref{RefinedDynamics} are fairly robust against noise, i.e., variations in $a_e(x,t)$ and $g_e(y)$. 

The rest of the paper is organized as follows.
In Section~\ref{Biological Background}, we review the biological background and related work.
In Section~\ref{Proofs}, we prove Theorem~\ref{main theorem}.
In Section~\ref{Bonifaci's Refined Model}, we prove Theorem~\ref{Bonifaci main theorem}. 
In Section~\ref{Open Problems}, we state some open problems.

\section{Background}\label{Biological Background}

\subsection{The Shortest Path Experiment}\label{Shortest Paths}

\begin{figure}[t]
\begin{center}
\includegraphics[width=0.4\textwidth]{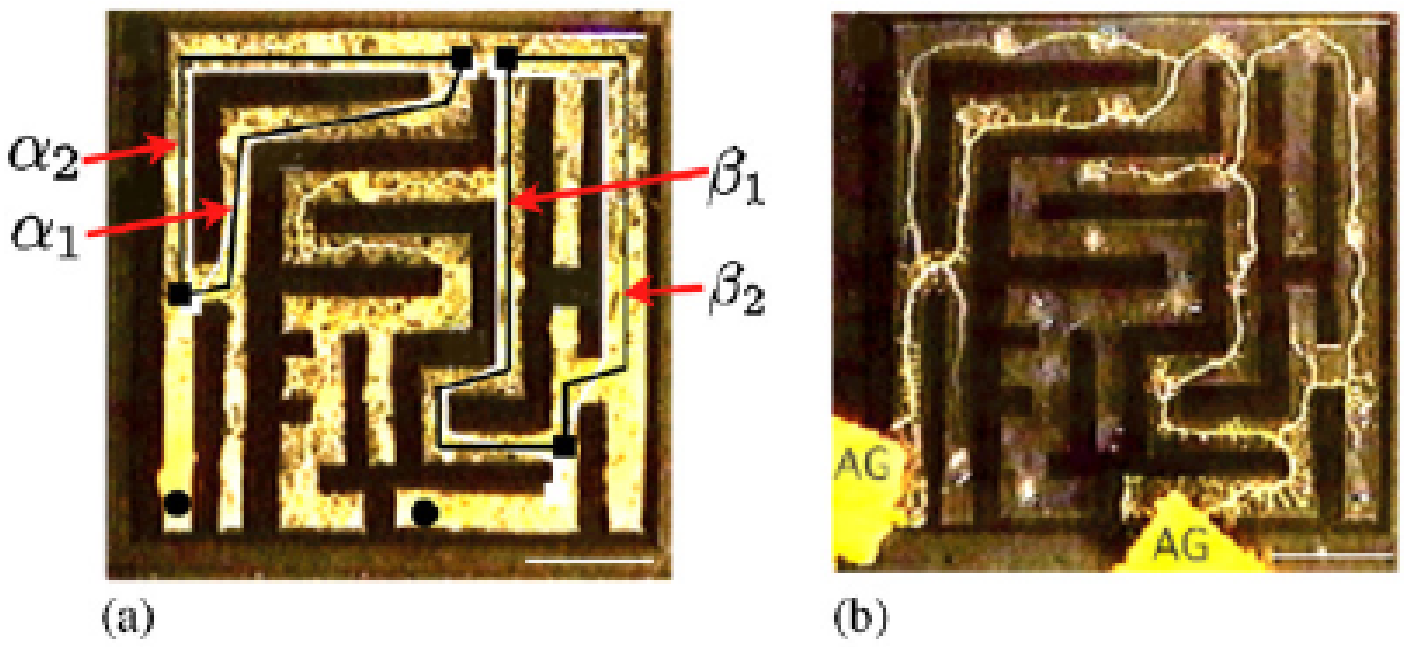}
\hspace{0.3cm}\vspace{-0.1em}
\includegraphics[width=0.4\textwidth]{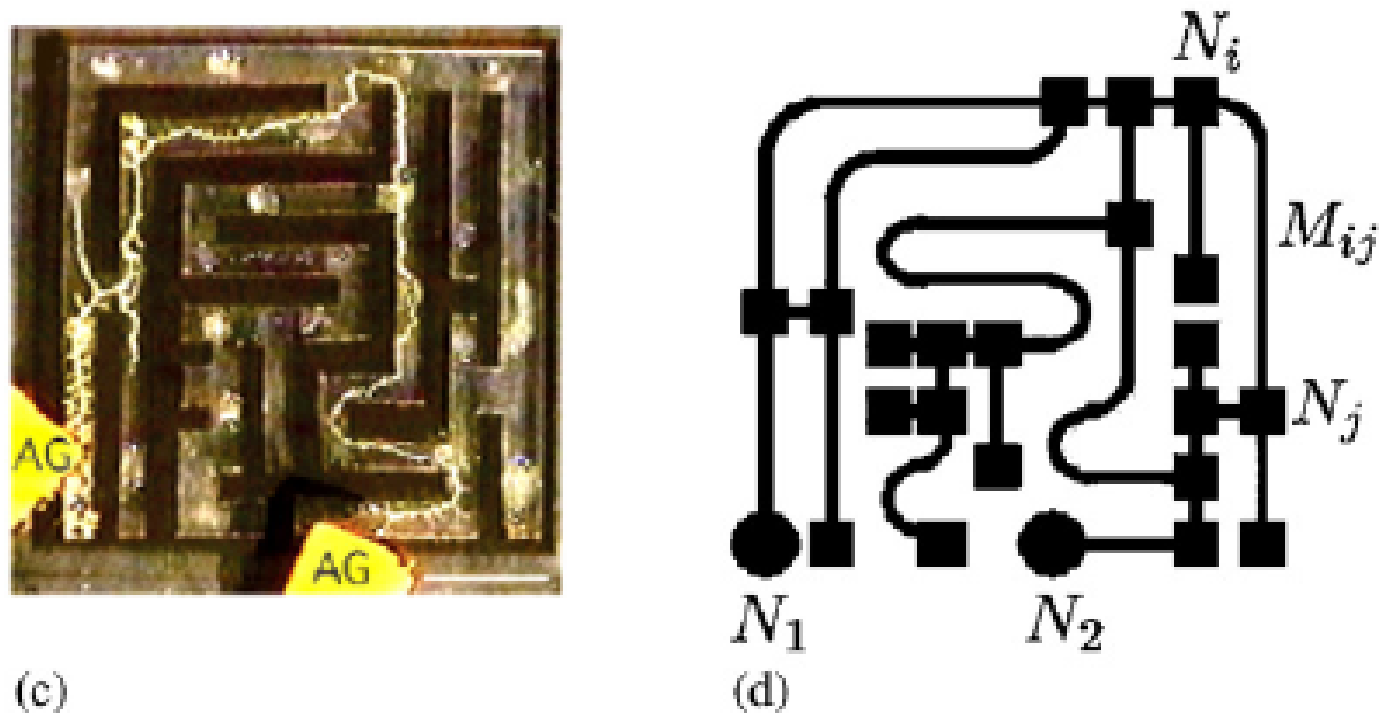}
\end{center}
\caption{\label{fig:maze} The experiment in~\cite{Nakagaki:2000}
(reprinted from there): (a) shows the maze uniformly covered by Physarum;
yellow color indicates presence of Physarum. Food (oatmeal) is provided at the
locations labeled AG. After a while the mold retracts to the shortest path
connecting the food sources as shown in (b) and (c). (d) shows the underlying
abstract graph. The video~\cite{Physarum-Video} shows the
experiment.
}
\end{figure}

\emph{Physarum polycephalum} is a slime mold that apparently is able to solve various optimization problems (see~\cite{PhysarumBook} for a survey of Physarum computations), in particular the shortest path problem.
Nakagaki, Yamada, and T\'{o}th~\cite{Nakagaki:2000} report about the following experiment; see Figure~\ref{fig:maze}. They built a maze, covered it by pieces of Physarum (the slime can be cut into pieces which will reunite if brought into vicinity), and then fed the slime with oatmeal at two locations. After a few hours the slime retracted to a path following the shortest path in the maze connecting the food sources. The authors report that they repeated the experiment with different mazes; in all experiments, Physarum retracted to the shortest path.

The paper~\cite{Tero-Kobayashi-Nakagaki} proposes a mathematical model for the behavior of the slime and argues extensively that the model is adequate. Physarum is modeled as an electrical network with time varying resistors. We have a simple \emph{undirected} graph $G = (N,E)$ with two distinguished nodes modeling the food sources. Each edge $e \in E$ has a positive length $c_e$ and a positive capacity $x_e(t)$; $c_e$ is fixed, but $x_e(t)$ is a function of time. The resistance $r_e(t)$ of $e$ is $r_e(t) = c_e/x_e(t)$. In the electrical network defined by these resistances, a current of value 1 is forced from one of the distinguished nodes to the other. For an (arbitrarily oriented) edge $e = (u,v)$, let $q_e(t)$ be the resulting current over $e$. Then, the capacity of $e$ evolves according to the differential equation  
\begin{equation}
\dot{x}_e(t) = | q_e(t) | - x_e(t),
\end{equation}
where $\dot{x}_e$ is the derivative of $x_e$ with respect to time. In equilibrium ($\dot{x}_e = 0$ for all $e$), the flow through any edge is equal to its capacity. In non-equilibrium, the capacity grows (shrinks) if the absolute value of the flow is larger (smaller) than the capacity.
% In the sequel, we will mostly drop the argument $t$ as is customary in the treatment of dynamical systems. We will also write $q$ for the vector with components $q_e$.
It is well-known that the electrical flow $q$ is the feasible flow minimizing energy dissipation $\sum_e r_e q_e^2$ (Thomson's principle).

\subsection{Minimum Risk Paths}\label{Minimum Risk Paths}
In~\cite{PhysarumMinimumRiskPath}, Nakagaki~et.~al.~study the following
scenario, see Figure~\ref{Fig:MinimumRiskPath}. They cover a rectangular plate with Physarum and feed it at opposite
corners of the plate. Two thirds of the plate is put under a bright light, one
third is kept in the dark. Under uniform lighting conditions, Physarum would
retract to a straight-line path connecting the food sources~\cite{Nakagaki:2000}. However, 
Physarum does not like light and therefore forms a path with
one kink connecting the food sources.  The path is such that the part under
light is shorter than in a straight-line connection. 
In the theory section of~\cite{PhysarumMinimumRiskPath}, the dynamics 
\begin{equation}\label{new dynamics}
\dot{x}_e(t) = 
\abs{q_e} - a_e x_e(t)
\end{equation}
is proposed. The constant $a_e$ is the \emph{decay rate} of edge $e$ if there is no flow on it. To model the experiment, $a_e = 1$
for edges in the dark part of the plate, and $a_e = C > 1$ for the edges in the
lighted area, where $C$ is a constant. Nakagaki et al.~\cite{PhysarumMinimumRiskPath}
report that in computer simulations, the
dynamics (\ref{new dynamics}) converges to the shortest source-sink path with
respect to 
the modified cost function $a_e c_e$. \ignore{ Observe that $a_e L_e = C L_e > L_e$ for the edges
in the lighted area, i.e., modified length is larger than Euclidean
length, and hence, a traveler of uniform speed with respect to the modified
length function seems to move slower in the lighted area than in the
dark. Hence the kink in the shortest path. }

\begin{figure}[ht!]
\centering{ \includegraphics[width=0.8\textwidth]{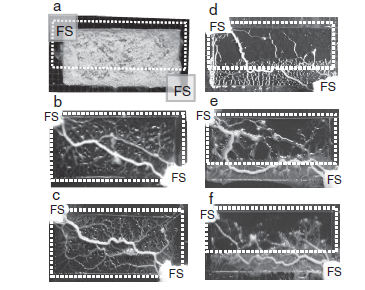}}
\caption{\label{Fig:MinimumRiskPath} Photographs of the connecting paths between two food
sources (FS). (a) The rectangular sheet-like morphology of the
organism immediately before the presentation of two FS and
illumination of the region indicated by the dashed white lines.
(b),(c) Examples of connecting paths in the control experiment
in which the field was uniformly illuminated. A thick tube was
formed in a straight line (with some deviations) between the FS.
(d)-(f) Typical connecting paths in a non-uniformly illuminated
field (95 K lx). Path length was reduced in the illuminated field,
although the total path length increased. Note that fluctuations in
the path are exhibited from experiment to experiment. (Figure and caption
reprinted from~\cite[Figure 2]{PhysarumMinimumRiskPath}.)}
\end{figure}

\subsection{A Reformulation: Nonuniform Physarum}

Let $y_e = a_e x_e$. The electrical flow $q$ is determined by the resistances $r_e =
c_e/x_e$. Therefore, we write  $q(r(t))$ instead of $q(t)$ for clarity. Next
observe that $r_e = c_e/x_e = (a_e c_e)/(y_e)$. Thus if we take $y$ as the
vector of edge capacities and $(a_ec_e)_e$ as the vector of costs, we get the
same electrical flow. We can express~\lref{new dynamics} as a dynamics for $y$ as
\[ \dot{y}_e = a_e \dot{x}_e = a_e  (\abs{q_e(r_e)} - a_e x_e) = a_e (\abs{q_e(r_e)} -
y_e).\]
So we may instead consider the dynamics 
\[  \dot{y_e} = a_e (\abs{q_e} - y_e) \]
under the modified cost function $a_e c_e$. This is our dynamics \lref{NUPD}, where we generalized further by allowing $a_e$ to depend on $x$ and $t$. 
In this model, the \emph{quantity $a_e$ indicates the responsiveness
  (reactivity) of an edge to
differences between flow and capacity}.  

\subsection{Beyond Shortest Paths}

The biological experiments concern shortest paths. The papers~\cite{Physarum,Bonifaci-Physarum} showed Theorem~\ref{uniform convergence} for the shortest path problem and the transportation problem; here $A$ is the node-arc incidence matrix of a directed graph, $b$ is the supply-demand vector of a transportation problem, i.e., $\sum_i b_i = 0$, and $c > 0$ are the edge costs. Convergence for the discretization of \lref{PD} was shown in~\cite{Physarum-Complexity-Bounds}.

The theoretical literature soon asked whether the dynamics~\lref{PD} can also solve more general problems. The basis pursuit problem was first studied in~\cite{SV-IRLS} and convergence of the discretization was shown. Theorem~\ref{uniform convergence} was shown in~\cite{BeckerBonifaciKarrenbauerKolevMehlhorn}. The function~\lref{Lyapunov Function A} was introduced in~\cite{Facca-Daneri-Cardin-Putti} and shown to be a Lyapunov function for \lref{PD} in~\cite{Facca-Cardin-Putti}. 

The paper~\cite{Bonifaci-RefinedModel} introduces and studies the dynamics~\lref{RefinedDynamics}.

\newcommand{\hQ}{\bar{q}}

The directed version of the Physarum dynamics evolves according to the differential equation 
\begin{equation}\label{directed dynamics}
  \dot{x}_e = q_e - x_e.
\end{equation}
No biological significance is claimed for this dynamics. It can solve linear programs with positive cost vectors~\cite{Ito-Convergence-Physarum,SV-LP}. In~\cite{Physarum-Complexity-Bounds}, convergence was claimed for the non-uniform dynamics $\dot{x}_e = a_e (q_e - x_e)$. The proof is incorrect~\cite{Physarum-Complexity-Bounds-Erratum}. 

% \footnote{
% In the unpublished proof of Lemma 10, the authors argue: ``We first follow the development in~\cite{Ito-Convergence-Physarum}, taking the
% reactivities into account. From
% \[    \frac{d}{ds} x_e e^{a_e s} = \dot{x_e} e^{a_e s} + a_e x_e e^{a_e
% s} = a_e (q_e - x_e) e^{a_e s} + a_e x_e e^{a_e s} = a_e q_e e^{a_e s}\]
% we have
% \[   x_e(t) e^{a_e t} - x_e(0) = \int_0^t a_e q_e(s) e^{a_e
%   s} ds \]
% and hence 
% \[    x_e(t) = x_e(0) e^{-a_e t} + \int_0^s a_e q_e(s) e^{-a_e
%   (t - s)} ds  = x_e(0) e^{-a_e t} +  (1 - e^{-a_e t})\int_0^s a_e q_e(s) \frac{e^{-a_e
%   (t - s)}}{1 - e^{-a_e t}}  ds. \]
% Let 
% \[   \hQ_e(t) = \int_0^s a_e q_e(s) \frac{e^{-a_e (t - s)}}{1
%   - e^{-a_e t}}  ds.\]
% Since $\int_0^t e^{-a_e (t - s)} ds = (1 - e^{-a_e t})/a_e$, $\hQ(t)$ is a convex
% combination of flows and hence a
% flow.''  This argument is fallacious. For each edge $\hQ_e$ is a convex combination. But these
% combinations are not uniform over edges. Therefore $\hQ$ is NOT a convex combination of flows.}

  \section{The Proof of Theorem~\ref{main theorem}}\label{Proofs}

  \subsection{Preliminaries}

  For a capacity vector $x \ge 0$ and a vector $f \in \R^m$ with $\supp(f) \subseteq \supp(x)$, we use
  \[ E_x(f) = \sum_e (c_e/x_e) f_e^2\]  to denote the \emph{energy} of $f$. 
  When $\supp(f) \not\subseteq \supp(x)$, the energy of $f$ is infinite. Further, we use
  \[ \cost(f)= \sum_e c_e \abs{f_e} = c^T \abs{f}\]  to denote the \emph{cost} of $f$. Note that
  \[ E_x(x) = \sum_e (c_e/x_e) x_e^2 = \sum_e c_e x_e = \cost(x). \]

We use $R$ to denote a diagonal matrix with entries $c_e/x_e$; here we use the convention that attention is restricted to the edges $e$ with $x_e > 0$. In part (a) of Theorem~\ref{main theorem}, it is shown that $x(t) > 0$ for all $t$ if $x(0) > 0$. However, in the limit some edges may have capacity zero. Energy-minimizing solutions are induced by node potentials $p \in \R^n$ according to the following equations:
\begin{align}
  b &= Aq \label{feasibility}\\
  q &= R^{-1} A^T p  \label{definition of q}\\
  AR^{-1} A^T p &=b      \label{definition of p}
\end{align}
We give a short justification. The vector $q$ minimizes the quadratic function $\sum_e (c_e/x_e) q_e^2$ subject to the constraints $Aq = b$. The KKT conditions (see~\cite[Subsection 5.5]{Boyd-Vandenberghe}) state that at the optimum, the gradient of the objective is a linear combination of the gradients of the constraints. Thus
\[       2 (c_e/x_e) q_e = \sum_i p_i A_{i,e}  \]
for some vector $p \in \R^n$. Absorbing the factor $2$ into $p$ yields equation \lref{definition of q}. Substitution of \lref{definition of q} into \lref{feasibility} gives \lref{definition of p}.

We next collect some well-known properties of the minimum energy solution; the proof of part (ii) can, for example, be found in~\cite{BeckerBonifaciKarrenbauerKolevMehlhorn}.
Let $D$ be the maximum absolute value of a square submatrix of $A$.
\begin{compactenum}[\mbox{}\hspace{\parindent}(i)]
\item The minimum energy solution is defined by~\lref{definition of q} and~\lref{definition of p}.
Moreover, it is unique.
\item $\abs{q_e} \le D\onenorm{b}$ for every $e \in [m]$. 
\item $E_x(q) = \sum_e (c_e/x_e) q_e^2 = b^T p$, where $p$ is defined by~\lref{definition of p}. This holds since \[E_x(q) = q^T R q = p^T A R^{-1} R R^{-1} A^T p = p^T A R^{-1} A^T p = p^T b.\]
\end{compactenum}
With the help of~\lref{definition of q}, the dynamics can we rewritten as
\begin{equation}\label{rewritten dynamics}  \dot{x}_e = a_e(x,t) \left(\left|\frac{x_e}{c_e} A_e^T p\right|- x_e\right) = a_e(x,t) x_e  \left( \frac{\abs{A_e^T p}}{c_e} - 1\right),\end{equation}
where $A_e$ denotes the $e$-th column of matrix $A$.

\subsection{Existence}\label{Existence}

The right-hand side of~\lref{NUPD} is locally Lipschitz-continuous in $x$ and $t$. The function $a_e(x,t)$ is locally Lipschitz by assumption, $q$ is an infinitely often differentiable rational function in the $x_e$ and hence locally Lipschitz. Furthermore, locally Lipschitz-continuous functions are closed under additions and multiplications.
Thus $x(t)$ is defined and unique for $t \in [0,t_0)$ for some $t_0$.

Since $a_e(x,t) \le C$ for all $e$, $x$ and $t$, we have $\dot{x}_e \ge -Cx$ and thus $x_e \ge x_e(0) e^{-Ct}$. Hence, $x(t) > 0$ for all $t$ and the solution does not reach the boundary of the domain in finite time.
Also since $\abs{q_e(t)} \le D\onenorm{b}$ for all $e$ and $t$, we have $\dot{x}_e \le C(D \onenorm{b} - x)$ and hence $x_e(t) \le \max(x_e(0), D \onenorm{b})$ for all $t$. In particular, the solution is bounded.  Thus, $t_0 = \infty$ by well-known results of maximal solutions of ordinary differential equations~\cite[Corollary 3.2]{Hartman}. 

The condition $a(x,t) \le C < \infty$ is crucial for existence. Let $n = 0$, $m = 1$ and $a(x,t) = 1/x$. The matrix $A$ is $0 \times 1$, i.e., there are no constraints. Then the minimum energy solution is the null-vector of dimension one and~\lref{NUPD} becomes $\dot{x} = 1/x \cdot (0 - x) = -1$; the domain of definition is $[0,x(0))$. 

\subsection{Fixed Points}

A point $x$ is a fixed point if $\dot{x} = 0$. In ~\cite{BeckerBonifaciKarrenbauerKolevMehlhorn} is was shown that the fixed points of~\lref{PD} are the vectors $\abs{f}$, where $f$ is a basic feasible solution of~\lref{LP}. This uses the assumption that any two basic feasible solutions have distinct cost. The proof carries over to~\lref{NUPD} under the additional assumption that $a_e(x,t) \ge \epsilon$ for all $e$, $x$ and $t$ and some positive $\epsilon$. Under this additional assumption $\dot{x} = 0$ is equivalent to $\abs{q} = x$ for~\lref{PD} and~\lref{NUPD}. 
This section is reprinted from~\cite{BeckerBonifaciKarrenbauerKolevMehlhorn} with minor adaptions.  A vector $f'$ is \emph{sign-compatible} with a vector $f$ (of the same dimension) if $f'_e \not= 0$ implies $f'_e f_e > 0$. In particular, $\supp(f') \subseteq \supp(f)$. We use the following corollary of the finite basis theorem for polyhedra.

\begin{lemma}\label{sign-compatible representation}
	Let $f$ be a feasible solution of \lref{LP}. Then $f$ is the sum of a convex combination of at most $n$ basic feasible solutions plus a vector in the kernel of $A$. Moreover, all elements in this representation are sign-compatible with $f$.
\end{lemma}
\begin{proof} We may assume $f \ge 0$. Otherwise, we flip the sign of the appropriate columns of $A$. Thus, the system $Af =b,\ f \ge 0$ is feasible and $f$ is the sum of a convex combination of at most $n$ basic feasible solutions plus a vector in the kernel of $A$ by the finite basis theorem~\cite[Corollary 7.1b]{Schrijver:Book}. By definition, the elements in this representation are non-negative vectors and hence sign-compatible with $f$.
\end{proof}

% We will first characterize the equilibrium points. They are precisely the points $\abs{f}$, where $f$ is a basic feasible solution; the proof uses property (B). We then show that $E_x(x)$ is a Lyapunov function for~\eqref{PD}, in particular, $\dot{E} \le 0$ and $\dot{E} = 0$ if and only if $x$ is an equilibrium point. For this argument, we need that the energy of $q$ is at most the energy of $x$ with equality if and only if $x$ is an equilibrium point. This proof uses (A) and (C). It follows from the general theory of dynamical systems that $x(t)$ approaches an equilibrium point. Finally, we show that convergence to a non-optimal equilibrium is impossible. 

\begin{lemma}\label{fixed points} Assume $a_e(x,t) \ge \epsilon$ for some positive $\epsilon$ and all $e$, $x$, and $t$, and that no two feasible solutions of $A f = b$ have the same cost.
	If $f$ is a basic feasible solution of~\eqref{LP}, then $x = \abs{f}$ is a fixed point. Conversely, if $x$ is a fixed point, then $x = \abs{f}$ for some basic feasible solution $f$.
\end{lemma}
\begin{proof}
	Let $f$ be a basic feasible solution, let $x = \abs{f}$, and let $q$ be the minimum energy feasible solution with respect to the resistances $c_e/x_e$. We have $Aq = b$ and $\supp(q) \subseteq \supp(x)$ by definition of $q$. Since $f$ is a basic feasible solution there is a subset $B$ of size $n$ of the columns of $A$ such that $A_B$ is non-singular and $f = (A_B^{-1} b, 0)$. Since $\supp(q) \subseteq \supp(x) =\supp(f)  \subseteq B$, we have $q = (q_B,0)$ for some vector $q_B$. Thus, $b = Aq = A_B q_B$ and hence $q_B = f_B$. Therefore $\dot{x} = \abs{q} - x = 0$ and $x$ is an fixed point.

	Conversely, if $x$ is an fixed point, $\abs{q_e} = x_e$ for every $e$. By changing the signs of some columns of $A$, we may assume $q \ge 0$. Then $q = x$. Since $q_e = (x_e/c_e) A_e^T p$ by~\lref{definition of q}, we have $c_e = A_e^T p$, whenever $x_e > 0$. By Lemma~\ref{sign-compatible representation}, $q$ is a convex combination of basic feasible solutions plus a vector in the kernel of $A$ that are sign-compatible with $q$. The vector in the kernel is zero since $q$ is a minimum energy solution\footnote{Assume $q = q^1 + q^2$ with $q^1 \ge 0$, $q_2 \ge 0$, $q_2 \not= 0$, and $A q_2 = 0$. Then $A q_1 = b$, $\supp(q_1) \subseteq \supp(q) \subseteq \supp(x)$, and $E_x(q_1) < E_x(q)$, a contradiction.}. For any basic feasible solution $z$ contributing to $q$, we have $\supp(z) \subseteq \supp(x)$.  Summing over the $e \in \supp(z)$, we obtain 
	\[ \cost(z) = \sum_{e \in \supp(z)} c_e z_e = \sum_{e \in \supp(z)} z_e A_e^T p = b^T p,\]
	i.e., all basic feasible solutions used to represent $q$ have the same cost. Since we assume the costs of distinct basic feasible solutions to be distinct, $q$ is a basic feasible solution. 
    \end{proof}

    \begin{corollary}\label{discrete set}  Assume $a_e(x,t) \ge \epsilon$ for some positive $\epsilon$ and all $e$, $x$, and $t$ and that no two feasible solutions of $A f = b$ have the same cost. Then the  set of fixed points is a discrete set. \end{corollary}

\subsection{The Lyapunov Function}

\begin{lemma}\label{Lyapunov Lemma}  $L(x,t) = p^T b + c^T x$ is a Lyapunov function for~\lref{rewritten dynamics}. More precisely, $\frac{d}{dt} L(x,t) \le 0$ always with equality only if  for all $e$ either $x_e = 0$ or $a_e(x,t) = 0$ or $\abs{A_e^Tp} = c_e$.  \end{lemma}
\begin{proof}
  Taking the derivative of \lref{definition of p} with respect to time yields
  \begin{equation}\label{helper} A \frac{d}{dt}(R^{-1}) A^T p + A R^{-1} A^T \dot{p} = 0 \end{equation}
  We next compute the derivative of both summands of $L(x,t)$ with respect to time separately. For the first summand we obtain
  \begin{align}
    \frac{d}{dt} p^T b &= \frac{d}{dt} p^T A R^{-1} A^T p = \dot{p}^T A R^{-1} A^T p +  p^T A \frac{d}{dt} (R^{-1}) A^T p + p^T A R^{-1} A^T \dot{p}\nonumber\\
                       &= - p^T A \frac{d}{dt} (R^{-1}) A^T p = - \sum_e \frac{(A_e^T p)^2}{c_e}\dot{x}_e \label{eq:dt_btp}\\
    &=  - \sum_e (A_e^T p)^2 \frac{a_e}{c_e} \left( \frac{x_e}{c_e} \abs{A_e^T p} - x_e \right)\nonumber\\
  &= - \sum_e a_e c_e x_e \left(\frac{\abs{A_e^T p}^3}{c_e^3} - \frac{\abs{A^T_e p}^2}{c_e^2}\right),\label{eq:dt_btpLem}
  \end{align}
where the first equality uses~\lref{definition of p}, the second equality follows from the product rule of differentiation, the third equality follows from~\lref{helper}, the fourth equality is a simple algebraic manipulation, the fifth equality follows from~\lref{rewritten dynamics}, and the last equality is a simple algebraic manipulation.

For the second summand, we obtain
\begin{align}
  \frac{d}{dt} c^T x = \sum_e c_e \dot{x}_e = \sum_e a_e c_e \left( \frac{x_e}{c_e} \abs{A_e^T p} - x_e \right)
  = \sum_e a_e c_e x_e \left( \frac{\abs{A_e^T p}}{c_e} - 1 \right).\label{eq:dt_ctx}
\end{align}
Combining \lref{eq:dt_btpLem} and \lref{eq:dt_ctx}, and writing $\lambda_e$ instead of 
$\abs{A_e^T p}/c_e$, yields
\[ \frac{d}{dt} L(x,t) = - \sum_e a_e c_e x_e \left(\lambda_e^3 - \lambda_e^2 - \lambda_e + 1\right). \]
Since
\[ \lambda_e^3 - \lambda_e^2 - \lambda_e + 1 = (\lambda_e^2 - 1)(\lambda_e - 1) = (\lambda_e + 1)(\lambda_e - 1)^2\]
and $\lambda_e \ge 0$, $\frac{d}{dt} L(x,t) \le 0$ always. Moreover, the derivative is equal to zero only if $a_e x_e (\lambda_e -1) = 0$ for all $e$, i.e., for all $e$ either $x_e = 0$ or $a_e(x,t) = 0$ or $\abs{A_e^Tp} = c_e$. 
\end{proof}

\begin{corollary}\label{fixed point = dot{L} = 0} Assume further $a_e(x,t) \ge \epsilon$ for some positive $\epsilon$ and all $e$, $x$ and $t$. Then $L(x,t) = 0$ if and only if $x$ is a fixed point. \end{corollary}
\begin{proof} We have $L(x,t) = 0$ if and only if for all $e$ either $x_e = 0$ or $\abs{A_e^Tp} = c_e$. The latter condition is equivalent to $\abs{q_e} = x_e/c_e \abs{A_e^Tp} = x_e$. Thus $\abs{q} = x$. \end{proof}

\subsection{Convergence}

From now on, we make the additional assumption that $a_e(x,t) \ge \epsilon$ for some positive $\epsilon$ and all $e$, $x$, and $t$. It then follows from the general theory of dynamical systems that $x(t)$ converges to a fixed point. 

\begin{corollary}[Generalization of Corollary 3.3.~in~\cite{Bonifaci-Physarum}.]\label{convergence to fixed point}
 Assume further $a_e(x,t) \ge \epsilon$ for all $e$, $x$ and $t$. As $t \rightarrow \infty$, $x(t)$ and $\abs{q(t)}$ approach a fixed point $x_0$. Moreover, $E_x(q)$ and $\cost(x)$ converge to $c^T x_0$. 
\end{corollary}
\begin{proof}
	The proof in~\cite{Bonifaci-Physarum} carries over. We include it for completeness. The existence of a Lyapunov function $L$ implies by \cite[Corollary 2.6.5]{LaSalle} that $x(t)$  approaches the set $\set{x \in \R_{\ge 0}^m}{\dot{L} = 0}$, which by Corollary~\ref{fixed point = dot{L} = 0} is the same as the set $\set{x \in \R_{\ge 0}^m}{\dot{x} = 0}$. Since this set consists of isolated points (Lemma~\ref{fixed points}), $x(t)$ must approach one of those points, say the point $x_0$. When $x = x_0$, one has $E_x(q) = E_x(x) = \cost(x) = c^T x$.
      \end{proof}

      The assumption $a_e(x,t) \ge \epsilon > 0$ is crucial as the following example shows. Let $n = m = 1$,  consider the task of minimizing $\abs{x}$ subject to the constraint $x = 1$, and let $a(x,t) = e^{-t}/2$ and $x(0) = 1/2$. Then
      $\dot{x} = e^{-t}(1 - x)/2$. Integrating from $0$ to $t$ and observing that 
      $x(t) \ge 1/2$ for all $t$, we obtain
      \[ x(t) = x(0) + \int_0^t \frac{e^{-s}}{2}(1 - x(s)) ds \le \frac{1}{2} + \frac{1}{4}
        \int_0^t e^{-s} ds \le \frac{1}{2} + \frac{1}{4} (1 - e^{-t}) \le \frac{3}{4} \]
      and hence the dynamics does not converge to the optimal solution $x^* = 1$, which, in this case, is the only fixed point. 

It remains to exclude that $x(t)$ converges to a non-optimal fixed point. We can do so under an additional assumption on $a(x,t)$.

\begin{theorem} 
  Assume further that $a_e(x,t)$ does not depend on $x$, i.e., $a_e(x,t) = a_e(t)$, $a_e(t) \ge \epsilon$ for some positive $\epsilon$ for all $e$ and $t$, and $\dot{a}_e(t)  \ge 0$ for all $e$ and $t$. 
	As $t \rightarrow \infty$, $x(t)$ converges to the optimal solution $x^*$. 
\end{theorem}
\begin{proof} 
	Assume that $x(t)$ converges to a non-optimal fixed point $z$. Let $x^*$ be the optimal solution, let $B$ be such that $x_e(t) \le B$ for all $e$ and $t$ (by Subsection~\ref{Existence} the solution is bounded), and let
	\[ W(x(t)) = \sum_e \frac{x^*_e c_e}{a_e} \ln \frac{x_e(t)}{B}. \]
	Let $\delta = (\cost(z) - \cost(x^*))/2$. 
	Then $E_x(q(t)) \ge \cost(z) - \delta =  \cost(x^*) + \delta$, for all sufficiently large $t$.
	Further, by definition $q_e = (x_e/c_e) A_e^T p$ and thus
	\[
		\dot{W}=\sum_{e}x_{e}^{*}c_{e}\frac{\abs{q_{e}}-x_{e}}{x_{e}} + \sum_{e}x_{e}^{*}c_{e}\left(\frac{-\dot{a}_{e}}{a_{e}^{2}}\right)\ln\frac{x_{e}}{B} \ge-\cost(x^{*})+\sum_{e}x_{e}^{*}\abs{A_{e}^{T}p}\geq\delta,
	\]
	where the first inequality follows from $\ln (x_e/B) \le 0$ and $\dot{a}_e \ge 0$,
	and the second inequality is due to
	\begin{equation}\label{eq:costxStarDelta}
		\sum_{e}x_{e}^{*}\abs{A_{e}^{T}p}\geq\sum_{e}x_{e}^{*}A_{e}^{T}p = b^{T}p = E_{x}(q)\geq\cost(x^{*})+\delta.
	\end{equation}
	Hence $W \rightarrow \infty$, a contradiction to the fact that $x$ is bounded.
\end{proof}

\section{Bonifaci's Refined Model}\label{Bonifaci's Refined Model}

Bonifaci~\cite{Bonifaci-RefinedModel} investigates the dynamics
\begin{equation*}            
	\dot{x}_e = x_e \left( g_e\left( \frac{\abs{q_e}}{x_e}\right) - 1 \right) \quad \text{for all $e \in [m]$}, 
\end{equation*}
      where the response function $g_e: \R_{\ge 0} \rightarrow \R_{\ge 0}$ is assumed to be an increasing differentiable function satisfying $g_e(1) = 1$. For the shortest path problem in a network of parallel links, Bonifaci shows convergence to an optimal solution. Bonifaci assumes the same response function for every edge, but his proof actually works for response functions depending on the edge. Concrete response functions of this type had been considered earlier in the literature:
      \begin{itemize}
      \item Non-saturating response: $g(y) = y^\mu$ for some $\mu > 0$.
      \item Saturating response: $g(y) = (1 + \alpha) y^\mu/(1 + \alpha y^\mu)$ for some $\mu, \alpha > 0$.
      \end{itemize}

      \begin{lemma} 
      	$L(x,t) = p^T b + c^T x$ is a Lyapunov function for the dynamics~\lref{RefinedDynamics}. Moreover $L(x,t) = 0$ if and only if $x$ is a fixed point of~\lref{RefinedDynamics}. 
      \end{lemma}
      \begin{proof} 
   	We proceed as in the proof of Lemma~\ref{Lyapunov Lemma}.
       Let $\lambda_e=\abs{A_e^T p}/c_e$ and note that $|q_e|/x_e=\lambda_e\geq0$. Then, we have   	
        \begin{align*}
          \frac{d}{dt} L(x,t) &= \frac{d}{dt} \left(p^T b + c^T x \right)
        \overset{\lref{eq:dt_btp}}{=}-\sum_{e}\left(\frac{(A_{e}^{T}p)^{2}}{c_{e}}-c_{e}\right)\dot{x}_{e}\\
        &= -\sum_{e}c_{e}x_{e}\left(\lambda_{e}^{2}-1\right)\left(g_{e}\left(\lambda_{e}\right)-1\right)\\
        & \le 0,
      \end{align*}
      where the inequality follows by $g_e(1) = 1$ and $g_e$ is an increasing function implies that the terms $\lambda_{e}^2 - 1$ and $g_e(\lambda_{e}) - 1$ have the same sign.

      Moreover, the derivative is zero if and only if for all $e$, either $x_e = 0$ or 
      $\lambda_{e}=1$ (as $\lambda_{e}\geq0$). 
      Since the latter condition is equivalent to $\abs{A_e^T p} = c_e$,
      it follows for every $e$ with $x_e \neq 0$ that 
      $\abs{q_e} = x_e$ or equivalently $\dot{x}_e = 0$.
    \end{proof}

    We remark that the proof above would even work for transfer-functions $g_e(x,t,y)$. 
    It is only important that $g_e(x,t,1) = 1$ and that the function is increasing in $y$. 

It now follows from the general theory of dynamical systems that $x(t)$ converges to a fixed point. 

\begin{corollary}[Generalization of Corollary 3.3. in~\cite{Bonifaci-Physarum}.]
As $t \rightarrow \infty$, $x(t)$ and $\abs{q}(t)$ approach a fixed point $x_0$. Moreover, $E_x(q)$ and $\cost(x)$ converge to $c^T x_0$. 
\end{corollary}
\begin{proof} Same proof as Corollary~\ref{convergence to fixed point}. \end{proof}

We finally show convergence to the optimum solution of~\lref{LP} under the additional assumption that $g_e(y) \ge 1 + \alpha(y - 1)$ for some $\alpha > 0$ and all $y$ and $e$.

\begin{theorem} 
  Assume further that $g_e(y) \ge 1 + \alpha(y - 1)$ for some $\alpha > 0$ and all $e$ and $y$.
	As $t \rightarrow \infty$, $x(t)$ converges to the optimal solution $x^*$. 
\end{theorem}
\begin{proof} 
	Assume that $x(t)$ converges to a non-optimal fixed point $z$. Let $x^*$ be the optimal solution and let \[W(x(t)) = \sum_e x^*_e c_e\ln x_e(t).\]
	Let $\delta = (\cost(z) - \cost(x^*))/2$. 
	Then $E_x(q(t)) \ge \cost(z) - \delta = \cost(x^*) + \delta$, for all sufficiently large $t$.
	Further, by definition $q_e = (x_e/c_e) A_e^T p$ and thus
	\begin{align*}
          \dot{W} &= \sum_e x^*_e c_e \frac{x_e \left( g_e\left(\frac{\abs{q_e}}{x_e}\right) - 1\right) }{x_e}
          \ge \sum_e x^*_e c_e \alpha\left(\frac{\abs{A_e^T p}}{c_e} - 1\right)\\
		&\ge \alpha \cdot \left( - \cost(x^*) + \sum_e x^*_e \abs{A_e^T p}  \right)
		 \geq \alpha\cdot \delta,
	\end{align*}
	where the first inequality follows from $f(y) \ge 1 + \alpha(y - 1)$ for all $y$ and the last inequality follows from \lref{eq:costxStarDelta}.
	Hence $W \rightarrow \infty$, a contradiction to the fact that $x$ is bounded.
\end{proof}

      We note that convex increasing functions satisfy $g(y) \ge g(1) + \alpha(y - 1)$ with $\alpha = g'(1)$.

      \section{Open Problems}\label{Open Problems}
      For the dynamics \lref{NUPD}, we showed convergence to the optimal solution 
      under the assumptions:
      \begin{compactenum}[\mbox{}\hspace{\parindent}(i)]
      	\item $a_e(x,t)$ is bounded and bounded away from zero;
      	\item $a_e(x,t) = a_e(t)$ does not depend on the state $x$;
      	\item $\dot{a}_e \ge 0$ always.
      \end{compactenum}
      We argued that assumption (i) is necessary. How about assumptions (ii) and (iii)?

      For the uniform dynamics, convergence of a suitable Euler discretization was shown in~\cite{Physarum-Complexity-Bounds, SV-IRLS} for the shortest path problem and the basis pursuit problem respectively. What can be said about the convergence of the discretization of the non-uniform dynamics?

      Our proof that Bonifaci's refined model converges to the optimum solution requires the additional assumption that $g_e(y) \ge 1 + \alpha(y - 1)$ for some $\alpha > 0$ and all $e$ and $y \ge 0$. Can this condition be relaxed?
      
      There is also the directed dynamics $\dot{x} = q - x$ considered in~\cite{Ito-Convergence-Physarum,SV-LP}. 
      Can convergence be shown for its non-uniform version? 
      This question is answered affirmatively in~\cite{Facca-Karrenbauer-Kolev-Mehlhorn:NonUniformPhysarum}.

 %      \bibliographystyle{alpha}
  %    \bibliography{ref,physarum}

      \renewcommand{\htmladdnormallink}[2]{#1}

\newcommand{\etalchar}[1]{$^{#1}$}

\end{document}